\documentclass[11pt,letter]{article}
\usepackage[margin=1in]{geometry}
\usepackage[utf8x]{inputenc}
\usepackage{amsmath, amsthm}
\usepackage{wrapfig,floatflt,graphicx,amssymb,textcomp,array,amsmath}
\usepackage{enumerate}
\usepackage{multirow}
\usepackage{tabularx}
\usepackage{color}
\usepackage{todonotes}
\usepackage[titletoc,title]{appendix}
\usepackage{url}
\usepackage[hidelinks]{hyperref}
\usepackage{lineno}

\newcommand{\etal}{{et~al.~\!}}

\title{Approximating bottleneck spanning trees on partitioned\\ tuples of points 
	\thanks{This research is supported by NSERC.}
}

\author{Ahmad Biniaz\thanks{School of Computer Science, University of Windsor, ahmad.biniaz@gmail.com}
	\and Anil Maheshwari\thanks{School of Computer Science, Carleton University, \{anil, michiel\}@scs.carleton.ca}
	\and  Michiel Smid\footnotemark[2]
}

\date{}

\newtheorem{corollary}{Corollary}

\newtheorem{theorem}{Theorem}

\newtheorem*{problem*}{Problem}

\newtheorem*{invariant*}{Invariant}

\begin{document}
	\maketitle
	\begin{abstract}
		
		We present approximation algorithms for the following NP-hard optimization problems related to bottleneck spanning trees in metric spaces.
		
		\begin{enumerate}
			
			\item  The {\em disjoint bottleneck spanning tree} problem: Given $n$ pairs of points in a metric space, find two disjoint trees each containing exactly one point from each pair and minimize the largest edge length (over all edges of both trees). It is known that approximating this problem by a factor better than 2 is NP-hard. We present a 4-approximation algorithm for this problem. This improves upon the previous best known approximation ratio of $9$. Our algorithm extends to a $(3k-2)$-approximation for a more general case where points are partitioned into $k$-tuples and we seek $k$ disjoint trees. 
			\item The {\em generalized bottleneck spanning tree} problem: Given $n$ points in some metric space that are partitioned into clusters of size at most 2, find a tree that contains exactly one point from each cluster and minimizes the largest edge length. We show that it is NP-hard to approximate this problem by a factor better than 2, and present a 3-approximation algorithm.

			\item  The {\em partitioned bottleneck spanning tree} problem: Given $kn$ points in some metric space, find $k$ trees each containing exactly $n$ points and minimize the largest edge length (over all edges of the $k$ trees). We show that it is NP-hard to approximate this problem by a factor better than 2 for any $k\geqslant 2$. We present an $\alpha$-approximation algorithm for this problem where $\alpha=2$ for $k=2,3$ and $\alpha=3$ for $k\geqslant 4$. Towards obtaining these approximation ratios we present tight upper bounds on the edge lengths of $k$ equal-size disjoint trees that can be obtained from the nodes of a given tree. This result is of independent interest. 
		\end{enumerate}
			
		 Our hardness proofs imply that it is NP-hard to approximate the non-metric version of the above problems within any constant factor. If we seek traveling salesperson tours (instead of trees) then our algorithms simply extend to achieve approximate solutions with factors three times those mentioned above. 
	\end{abstract}
	
\section{Introduction}
Spanning tree is a fundamental structure in graph theory and combinatorics.  
The problem of finding spanning trees with enforced properties has received considerable attention from both theoretical and practical points of view. For example, the minimum spanning tree (MST) problem asks for a spanning tree with minimum total edge-length, and the bottleneck spanning tree (BST) problem asks for a spanning tree whose largest edge-length is minimum. Beside their interesting theoretical properties, these problems find applications in the design of networks, including computer networks, wireless networks, and transportation networks, to name a few. Bottleneck spanning trees in particular are important in designing telecommunications networks with short connections (edges). Short connections are desirable in many ways because they require lower transmission ranges, are more secure, and cause less interference. 
This paper addresses three closely related bottleneck spanning tree problems (illustrated in Figure~\ref{intro-fig}):

\begin{figure}[htb]
	\centering
	\setlength{\tabcolsep}{0in}
	$\begin{tabular}{ccc}
		\multicolumn{1}{m{.33\columnwidth}}{\centering\includegraphics[width=.18\columnwidth]{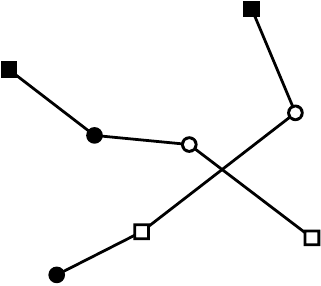}}
		&\multicolumn{1}{m{.33\columnwidth}}{\centering\includegraphics[width=.18\columnwidth]{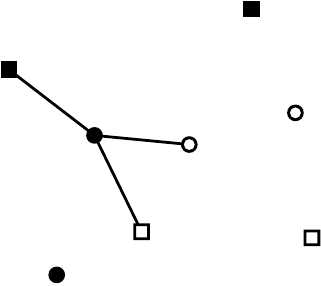}}
		
		&\multicolumn{1}{m{.33\columnwidth}}{\centering\includegraphics[width=.18\columnwidth]{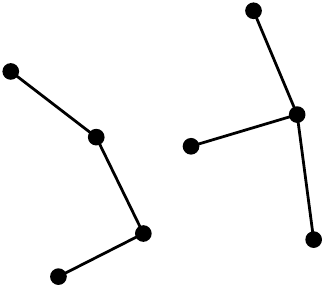}}
		\\
		(a) 2-DBST &(b) 2-GBST&(c) 2-PBST
	\end{tabular}$
	\caption{Illustration of the problems for $k=2$; black and white squares/circles represent tuples.}
	\label{intro-fig}
\end{figure}
\vspace{-8pt}

\begin{enumerate}
	
	\item  The {\em disjoint bottleneck spanning tree} ($k$-DBST) problem: Given $kn$ points in some metric space that are partitioned into $k$-tuples, find $k$ disjoint trees each containing exactly one point from each tuple and minimize the largest edge length (over all edges of the $k$ trees). 
	
	\item The {\em generalized bottleneck spanning tree} ($k$-GBST) problem: Given $n$ points in some metric space that are partitioned into clusters of size at most $k$, find a tree that contains exactly one point from each cluster and minimizes the largest edge length. The term ``spanning'' refers to span all clusters.

	\item  The {\em partitioned bottleneck spanning tree} ($k$-PBST) problem: Given $kn$ points in some metric space, find $k$ trees each containing exactly $n$ points and minimize the largest edge length (over all edges of the $k$ trees). 
\end{enumerate}

The above problems are natural generalizations of the standard BST problem. For $k=1$, all above problems are equivalent to the BST problem which can be solved optimally in polynomial time \cite{Camerini1978}. For $k\geqslant 2$, all above problems are NP-hard and cannot be approximated by a factor better than 2 unless P = NP (this will become clear shortly). The focus of this paper is on $k\geqslant 2$. We first present constant-factor approximation algorithms for $k=2$. Then we extend some of our algorithms for larger $k$.

\subsection{Some related works and applications}

The problems introduced above find real-world applications that we put into context together with some related works. In our description we implicitly assume that $k$ is at least 2.

\paragraph{(1)} The $k$-DBST problem is introduced by Arkin~\etal\cite{Arkin2017}. Motivated by the problem of maintaining secure connectivity in networks involving replicated data, Arkin~\etal\cite{Arkin2017} introduced a class of problems that ask for $k$ disjoint structures (trees, cycles, matchings) each containing one point form every given $k$-tuple.  
In particular they studied these problems for $k=2$.
Among many interesting results they presented a 9-approximation algorithm for the 2-DBST problem and an 18-approximation algorithm for computing two disjoint traveling salesperson tours (instead of trees). It is easily seen, from their Lemma 8, that the 9-approximation algorithm can be extended to achieve a $(6k-3)$-approximation for the $k$-DBST problem.  Although some of the results of Arkin~\etal\cite{Arkin2017} have been improved by Johnson \cite{Johnson2018}, their ratios 9 and 18 are still the best known. As for the lower bound, Johnson \cite{Johnson2018} showed that it is NP-hard to approximate the 2-DBST problem by a factor better than 2. 

\paragraph{(2)} The $k$-GBST problem is closely related to the $k$-generalized minimum spanning tree ($k$-GMST) problem, introduced by Myung~\etal\cite{Myung1995}. The $k$-GMST problem asks for a tree that contains exactly one point from each cluster and minimizes the total-edge length. This problem is well studied (see e.g. the recent survey by Pop~\cite{Pop2020} and references therein). The $k$-GMST problem is NP-hard even for $k=2$ in the Euclidean plane. Even a more restricted version where the two points in each cluster have the same $x$ or $y$ coordinates is NP-hard  \cite{Dey2021,Fraser2012,Ataei2018}. The metric version of the $k$-GBST can be approximated by a ratio of $2k$ using linear programming \cite{Pop2005} combined with the so-called parsimonious
property \cite{Goemans1993}. Related work \cite{Bhattacharya2015,Pop2020,Pop2001} also addresses the generalized traveling salesperson problem (TSP) in which the tour must contain exactly one point from each cluster.
The group Steiner tree is another related problem which asks for a shortest tree that contains {\em at least} one point from each cluster. The non-metric versions of both the $k$-GMST and the group Steiner tree  problems are NP-hard and cannot be approximated within any constant factor \cite{Halperin2003,Myung1995}.
Gabow~\etal\cite{Gabow1976} studied the problem of finding a path, from a source to a destination in a graph, that passes through at most one vertex from every given pair of vertices. Arkin~\etal\cite{Arkin2000} studied the multiple-choice minimum-diameter problem which is to select at least one element from each cluster to minimize the diameter of the chosen elements.
The $k$-GBST 
also lies in the concept of {\em imprecision} in computational geometry where each input point is provided as a region of uncertainty (also known as {\em neighborhood}) and the exact position of the point may be anywhere in the region; see e.g. \cite{Blanco2017,Dorrigiv2015,Loffler2009,Mitchell2007, Mitchell2010}.

Both the $k$-GBST and the $k$-GMST have real-world applications for example in the field of telecommunications, designing metropolitan area networks, interconnecting local area networks, determining location of regional service centers (e.g., stores, warehouses, agricultural settings, distribution centers). For a detailed explanation of these applications and for more examples we refer the interested reader to the paper of Myung~\etal\cite{Myung1995} and the recent survey by Pop~\cite{Pop2020}.

\paragraph{(3)} The $k$-PBST problem falls in the class of partitioning a set into subsets such that the substructures (computed on subsets) are balanced. Balanced partitioning of the input has a long history and gives rise to interesting theoretical problems. For example in the $k$-partition traveling salesperson problem we are given $k$ salespersons and the goal is to visit every city by exactly one salesperson and minimize the distance traveled by the salesperson making the longest journey \cite{Averbakh1996,Averbakh1997,Nagamochi2004}. 

The problem of $k$-balanced partitioning of a graph asks for partitioning the vertices of the graph into $k$ subsets such that the induced subgraph on each subset is connected and the maximum cardinality of the subsets is minimized. Dyer and Frieze \cite{Dyer1985} showed that this problem is NP-hard; they also showed the hardness of many variations of this problem. Chleb{\'{\i}}kov{\'{a}} \cite{Chlebikova1996} presented constant-factor approximations for $k=2,3$, and Chen~\etal\cite{Chen2019} presented a $k/2$-approximation for $k\geqslant 4$. The max-min version of this problem is also studied \cite{Chlebikova1996, Wakabayashi2007}.

Motivated by a problem from the shipbuilding industry, Andersson~\etal\cite{Andersson2003} studied the $k$-partition minimum spanning tree ($k$-PMST) problem where the goal is to partition an input point set into $k$ subsets such that the length of the longest MST on the subsets is minimized. As noted in \cite{Karakawa2011} (and references therein) the $k$-PMST problem also arises in multi-vehicle scheduling, task sequencing, and political districting. Andersson~\etal\cite{Andersson2003} showed that the $k$-PMST problem is NP-hard even for $k=2$ in the Euclidean metric in the plane, and presented $(4/3+\epsilon)$ and $(2+\epsilon)$ approximations for $k=2$ and $k\geqslant 3$, respectively. Karakawa~\etal\cite{Karakawa2011} studied this problem in higher dimensions. The $k$-PMST problem has also been studied in trees and cactus graphs under the name ``minmax subtree cover'' problem \cite{Nagamochi2004b,Nagamochi2006,Nagamochi2003}.

\subsection{Our contributions}
We study the $k$-DBST, $k$-GBST, and $k$-PBST problems in metric spaces (where distances satisfy the triangle inequality). We show the hardness as well as approximation algorithms for these problems. We present our results for the simplest version where $k=2$ (as it is easier to understand) and then extend them for larger $k$.
\begin{itemize}
\item The $2$-DBST problem is known \cite{Johnson2018} to be NP-hard and inapproximable by a factor better than 2. We present a 4-approximation algorithm for this problem. This improves the previous best known ratio of $9$ due to Arkin~\etal\cite{Arkin2017}. We extend our algorithm and achieve a $(3k-2)$-approximation for the $k$-DBST for any $k\geqslant 2$ (Theorem~\ref{kDBST-thr}).

\item The difficulty of the $2$-GBST problem lies in choosing representative points from clusters; once these points are selected, the problem is reduced to the standard BST problem. We show that it is NP-hard to approximate the $2$-GBST problem by a factor better than 2 using a reduction from 3-SAT (Theorem~\ref{2GBST-hardness}), and present a 3-approximation algorithm for this problem (Theorem~\ref{2GBST-thr}). In some part of our algorithm we show the following result which is of independent interest (Theorem~\ref{T1T2-thr}): Given a tree $T_1$ and a partitioning of its nodes into clusters of size at most two, we can obtain a tree $T_2$ that contains exactly one node from each cluster and the length of its edges is at most $3$ in the metric\footnote{In this metric the distance between two nodes $u$ and $v$ in a tree $T$ is the number of edges in the unique path between them in $T$.} of $T_1$; the upper bound $3$ is the best achievable.

\item We show that it is NP-hard to approximate the $k$-PBST problem by a factor better than 2 for any $k\geqslant 2$ (Theorem~\ref{2PBST-hardness}) using a reduction from the $2$-balanced partitioning of a graph \cite{Dyer1985}. 
We present an $\alpha$-approximation algorithm for this problem (Theorem~\ref{PBST-thr}) where $\alpha=2$ for $k=2,3$ and $\alpha=3$ for $k\geqslant 4$. Towards obtaining these approximation ratios we present tight upper bounds on the edge lengths of $k$ equal-size disjoint trees that can be obtained from the nodes of a given tree (Theorem~\ref{partitioning-thr}). This result is of independent interest. 
\end{itemize}

A straightforward implication of our hardness proofs and that of Johnson~\cite{Johnson2018} is that the non-metric versions of the above problems cannot be approximated within any constant factor. 
\paragraph{Extension to bottleneck TSP tours.}
If instead of trees in the above problems we seek TSP tours that minimize the largest edge length, then our algorithms simply extend to obtain approximate solutions with factors that are thrice those for bottleneck trees. This can be done via a known result that the {\em cube}\footnote{The cube of a graph $G$ has the same vertices as $G$, and has an edge between two distinct vertices if and only if there exists a path, with at most three edges, between them in $G$.} of every connected graph has a Hamiltonian cycle, and such a cycle can be computed in polynomial time \cite{Karaganis1968,Lesniak1973}; this is also hinted in \cite[Exercise 37.2.3]{Cormen1990}. 
To use this result, we first obtain an $\alpha$-approximate solution, namely $\cal B$, for the corresponding BST problem (using our BST algorithms) and then we find TSP tours, namely $\cal T$, in the cube of $\cal B$. By the triangle inequality the largest edge-length in the cube graph, and in particular in $\cal T$, is at most thrice the largest edge-length in $\cal B$. Notice that in all above problems the largest edge length in any optimal BST solution is a lower bound for the largest edge length in any optimal TSP solution. Thus $\cal T$ would be a $3\alpha$-approximate solution for the TSP. For example our $4$-approximation algorithm for the $2$-DBST can be extended to obtain a 12-approximation for two disjoint TSP tours that minimize the largest edge length; this improves the previous approximation ratio of 18 due to Arkin~\etal\cite{Arkin2017}. 
\paragraph{Notation.} The largest edge length in a tree $T$ is referred to as the {\em bottleneck} of $T$ and is denoted by $\lambda(T)$. We denote the distance between two points $p$ and $q$ in a metric space by $|pq|$.  Conceptually, a point set $P$ in a metric space can be viewed as a {\em metric graph}, i.e., as a complete edge-weighted graph with vertex set $P$ where the weight $w(e)$ of each edge $e=(p,q)$ is equal to the distance between $p$ and $q$, that is $w(e)=|pq|$.

\section{The $k$-DBST problem} 

Let $k\geqslant 2$ be an integer. In this section we present an approximation algorithm for the $k$-DBST problem: Given $kn$ points in some metric space that are partitioned into $k$-tuples, we want to find $k$ disjoint trees each containing exactly one point from each tuple and minimize the largest edge length (over all the $k$ trees).
We first present our approximation algorithm for $k=2$ as it is easier to understand. Then we extend the algorithm to larger $k$. 
Our algorithm benefits from the following remarkable result of K\"{o}nig which is stated in \cite{Hall1935}.

\begin{theorem}[K\"{o}nig, 1916]
	\label{Konig-thr}
	Let $S$ be any set with $kn$ elements that is partitioned, in two different ways, into $n$ subsets each with $k$ elements, namely $A_1,\dots,A_n$ and $B_1,\dots,B_n$. Then there exist $n$ elements of $S$, namely $r_1,\dots,r_n$, and a permutation $\pi$ of $\{1,\dots,n\}$ such that  $r_i\in A_i\cap B_{\pi(i)}$ for all $i\in\{1,\dots,n\}$.
\end{theorem}

\begin{center}
	\fbox{%
	\parbox{.85\textwidth}{%
		\vspace{3pt}
		{\bf Example.} Let $k=3$, $n=4$, $S=\{1,2,\dots,12\}$, and consider two partitions of $S$
		\vspace{-8pt} 
		$$A_1=\{1,2,3\},~ A_2=\{4,5,6\},~ A_3=\{7,8,9\},~ A_4=\{10,11,12\}$$
		\vspace{-20pt}  $$B_1=\{4,9,12\},~ B_2=\{2,8,11\},~ B_3=\{1,3,5\},~ B_4=\{6,7,10\}.$$ Then by taking $r_1=1$, $r_2=6$, $r_3=8$, $r_4=12$, and $\pi=(3,4,2,1)$ we get that
		\vspace{-8pt}  
		\[r_1\in A_1\cap B_3,~ r_2\in A_2\cap B_4,~ r_3\in A_3\cap B_2,~ r_4\in A_4\cap B_1.\vspace{-8pt}\]	}
}
\end{center}

Hall (1935) showed a more general version of K\"{o}nig's theorem (where subsets can have different sizes) as an implication of his famous result \cite{Hall1935}---today known as the Hall's marriage theorem. The set $R=\{r_1,\dots, r_n\}$ in Theorem~\ref{Konig-thr} is called a {\em complete system of representatives} for subsets $A_i$ (and also for subsets $B_i$).  
The following theorem (which is a generalized version of Lemma 8 in \cite{Arkin2017}) is an implication of K\"{o}nig's theorem.

\begin{theorem}
	\label{labeling-thr}
	Let $S$ be a set with $kn$ elements that is partitioned, in two different ways, into $n$ subsets each with $k$ elements, namely $A_1,\dots,A_n$ and $B_1,\dots,B_n$. Then, it is possible to label all elements of $S$ with $k$ distinct labels such that the $k$ elements in each of $A_1,\dots,A_n,B_1,\dots,B_n$ have $k$ distinct labels. Moreover, such a labeling can be found in polynomial time.
\end{theorem}

\begin{proof}
	By K\"{o}nig's theorem there exists a subset $R=\{r_1,\dots,r_n\}$ of $S$ that is a complete system of representatives for subsets $A_i$ and for subsets $B_i$. Such a system $R$ can be found as follows. Construct a bipartite graph $G= (V,E)$ with $2n$ vertices such that $V=\{A_1,\dots,A_n, B_1,\dots,B_n\}$ and there is an edge between $A_i$ and $B_j$ if and only if $A_i\cap B_j\neq \emptyset$. According to Hall's marriage theorem \cite{Bondy1976,Hall1935} $G$ has a perfect matching $M$ (with $n$ edges) which can be found in polynomial time. For every edge $(A_i,B_j)$ in $M$ pick an arbitrary representative element in $A_i\cap B_j$. These $n$ representatives form $R$.
	
	Label all elements of $R$ by $l_1$. Then remove the vertices of $R$ from $S$ and from corresponding subsets $A_i$ and $B_j$. As a result we obtain a new set $S$ with $(k-1)n$ elements and two distinct partitions of $S$ each with $n$ subsets of size $k-1$. By applying K\"{o}nig's and Hall's theorems we can find another complete system of representatives, and label them $l_2$. Repeating the above process achieves a desired labeling $l_1,\dots, l_k$.
\end{proof}

In the example above we can label elements of $S$ by $k~(=3)$ labels $l_1,l_2,l_3$ where (with a slight abuse of notation) $l_3=\{1,6,8,12\}$, $l_2=\{2,5,9,10\}$, and $l_3=\{3,4,7,11\}$
such that all elements in each $A_i$ and $B_i$ have different labels.

\subsection{A $4$-approximation for the $2$-DBST}
\label{2-BST-section}

In this section we present a 4-approximation algorithm for the $2$-DBST problem. 
Let $P$ be a set of $2n$ points in a metric space that is partitioned into $n$ tuples $A_1,\dots, A_n$ each with two points. Let $\lambda^*$ denote the bottleneck of a fixed optimal solution (consisting of two trees). 
We show how to find two disjoint trees $R$ and $B$ with edges of length at most $4\lambda^*$. To simplify our description we assume that the nodes of $R$ and $B$ are colored red and blue, respectively.

We start by computing a minimum spanning tree of $P$, which is also a bottleneck spanning tree. Let $e$ be a longest edge of $T$, that is $\lambda(T)=w(e)$. Let $T_1$ and $T_2$ be the two trees obtained by removing $e$ from $T$. Notice that $\max\{\lambda(T_1),\lambda(T_2)\}\leqslant w(e)$. If each $A_i$ has a point in $T_1$ and a point in $T_2$, then we claim that $R=T_1$ and $B=T_2$ form an optimal solution because if the fixed optimal solution contains an edge between a node of $T_1$ and a node of $T_2$ then the length of that edge is at least $w(e)$ which implies that $\lambda^*\geqslant w(e)$. Therefore $\max\{\lambda(R),\lambda(B)\}\leqslant \lambda^*$.

Now assume that both points of some tuple $A_i$ belong to say $T_1$. In any feasible solution, one point of $A_i$ is red and the other is blue. Then regardless of the coloring of the nodes of $T_2$, the optimal solution should contain an edge between a node of $T_1$ and a node of $T_2$. Thus $\lambda^*\geqslant w(e)$. We are going to color the nodes of $T$ (which are the points of $P$) red and blue and then obtain $R$ and $B$ in such a way that $\max\{\lambda(R),\lambda(B)\}\leqslant 4\cdot \lambda (T)$. This will imply that $\max\{\lambda(R),\lambda(B)\}\leqslant 4 \lambda^*$.

\begin{figure}[htb]
	\centering
	\setlength{\tabcolsep}{0in}
	$\begin{tabular}{cc}
		\multicolumn{1}{m{.6\columnwidth}}{\centering\includegraphics[width=.51\columnwidth]{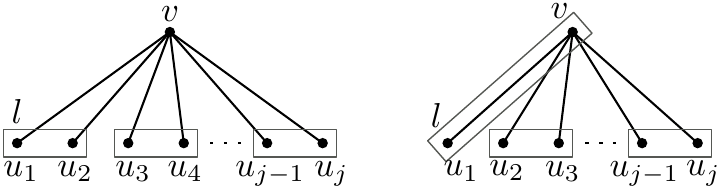}}
		&\multicolumn{1}{m{.4\columnwidth}}{\centering\includegraphics[width=.26\columnwidth]{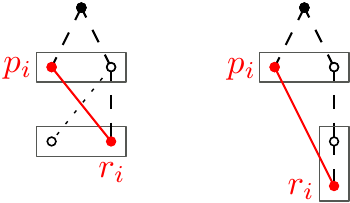}}
		\\
		(a)&(b)
	\end{tabular}$
	\caption{(a) Creating buckets. (a) Construction of $R$; dashed edges represent $\delta_i$.}
	\label{coloring-fig}
\end{figure}

We root $T$ at a leaf $q$. Then we partition the nodes of $T$ into $n$ buckets $B_1,\dots,B_n$ each with two vertices. The partitioning is done iteratively in a bottom-up fashion as follows. Consider a deepest leaf $l$ and let $v$ be its parent. Let $u_1,u_2,\dots,u_j$ be the children of $v$ where $u_1=l$ as in Figure~\ref{coloring-fig}(a). If $j$ is even then we create $j/2$ buckets $\{u_1,u_2\}, \{u_3,u_4\}, \dots,\{u_{j-1},u_j\}$, and then remove $u_1,\dots, u_j$ from $T$. If $j$ is odd then we create $(j+1)/2$ buckets $\{v,u_1\}$, $\{u_2,u_3\}, \{u_4,u_5\}, \dots,\{u_{j-1},u_j\}$, and then remove $v,u_1,\dots, u_j$ from $T$. Then we repeat the above process until $q$ and its only child form a bucket. We denote this last bucket by $B_n$. The total number of buckets is $n$ because $T$ has $2n$ nodes initially. Between any two nodes in the same bucket there exists a path of length at most $2$ in $T$, because the two nodes are either siblings or a child and its parent.

Now that we have two partitions $A_1,\dots,A_n$ and $B_1,\dots,B_n$ of $P$, we color (or label) the points of $P$ by two colors, red and blue, as in Theorem~\ref{labeling-thr}. Thus in each $A_i$ and each $B_i$ we get a red point and a blue point. We construct the tree $R$ by interconnecting the red points of buckets as follows; see Figure~\ref{coloring-fig}(b): Consider each bucket $B_i$ with $i\in\{1,\dots,n-1\}$ and let $r_i$ denote its red point. 
\begin{itemize}
	\item[(i)] If the parent of $r_i$ is not in $B_i$, then we connect $r_i$ to the red point of its parent's bucket.
	\item[(ii)] If the parent of $r_i$ is in $B_i$, then we connect $r_i$ to the red point of its grandparent's bucket.
\end{itemize}
We construct the tree $B$ on the blue points in a similar fashion. We claim that $R$ and $B$ are the desired trees. Since each $A_i$ contains a red point and a blue point (by Theorem~\ref{labeling-thr}), each of $R$ and $B$ contains exactly one point from $A_i$. Thus $R$ and $B$ form a feasible solution for the problem.
\paragraph{Analysis of the approximation ratio.} We show that $\lambda(R)\leqslant 4\cdot \lambda(T)$; an analogous argument holds for $B$. Root $R$ at the red point of $B_n$. Consider any red node $r_i$ in $R$ where $i\in\{1,\dots,n-1\}$.  Recall that $r_i\in B_i$. Let $p_i$ be the parent of $r_i$ in $R$. It suffices to show that $|r_ip_i|\leqslant 4\cdot \lambda(T)$.  Consider the unique path $\delta_i$ between $r_i$ and $p_i$ in $T$. See Figure~\ref{coloring-fig}(b). If $r_i$ was connected to $p_i$ in step (i) then $\delta_i$ has at most $3$ edges. If $r_i$ was connected to $p_i$ in step (ii) then $\delta_i$ has at most $4$ edges. Therefore $|r_ip_i|\leqslant w(\delta_i)\leqslant 4\cdot \lambda(T)$.

\subsection{A $(3k-2)$-approximation for the $k$-DBST}
Here we extend our 4-approximation algorithm of the previous section to get a $(3k-2)$-approximation for the $k$-DBST problem. We should note that (although it is not mentioned explicitly in their paper) Theorem 7 from Arkin~\etal\cite{Arkin2017} combined with their Lemma 8 already gives a  $(6k-3)$-approximation for the $k$-DBST problem. 

Let $P$ be a set of $kn$ points that is partitioned into $n$ tuples $A_1,\dots, A_n$ each with $k$ points. Let $\lambda^*$ denote the bottleneck of a fixed optimal solution (consisting of $k$ trees). 
We show how to color the points in each $A_i$ by $k$ colors $c_1,\dots,c_k$, and to obtain a tree $T_i$ on all points with color $c_i$ such that  $\lambda(T_i)\leqslant (3k-2)\lambda^*$.

Let $T$ be a minimum spanning tree of $P$. Root $T$ at a leaf $q$. We partition the nodes of $T$ into $n$ buckets $B_1,\dots,B_n$ each with $k$ nodes. The partitioning is done iteratively in a bottom-up fashion. We describe it for obtaining bucket $B_j$. For each node $v$ in the current tree $T$, let $N(v)$ denote the number of nodes in the subtree rooted at $v$, including $v$ itself. Then we look at all nodes $v$ for which $N(v)$ is at least $k$. Among those, pick a node $v$ for which $N(v)$ is minimum. 
Then $N(v)$ is at least $k$ and each of its children has a subtree of size at most $k-1$.
Now we make $B_j$: Take a leaf in the subtree of $v$, add it to $B_j$, and remove it from the tree. Repeat this until $B_j$ has size $k$. 

With the two partitions $A_1,\dots,A_n$ and $B_1,\dots,B_n$ in hand, we color the points of $P$ by $k$ colors $c_1,\dots,c_k$ as in Theorem~\ref{labeling-thr}. Thus in each $A_i$ and in each $B_i$ we get $k$ distinct colors.

\begin{wrapfigure}{r}{0.19\textwidth}
	\begin{center}
		\vspace{-22pt}
		\includegraphics[width=.15\textwidth]{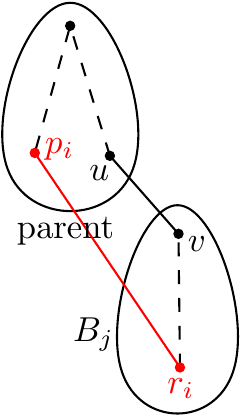}
		\label{notation-fig}
	\end{center}
	\vspace{-14pt}
\end{wrapfigure}

Notice that between any two points in the subtree of $v$ there is a path in $T$ with at most $2k-2$ edges. We say that $v$ is the {\em representative} of $B_j$. Moreover, we define the {\em parent} of $B_j$ to be the bucket containing $v$ (if $v\notin B_j$) or the bucket containing $v$'s parent (if $v\in B_j$). 
For each color $c_i$ we construct $T_i$ as follows: for each bucket $B_j$ we connect its point with color $c_i$ (say point $r_i$) to the point with color $c_i$ in $B_j$'s parent bucket (say point $p_i$). To prove the approximation ratio it suffices to show that between $r_i$ and $p_i$ there is a path of length at most $3k-2$ in $T$. This is easily seen as there is a path of length at most $k-1$ from $r_i$ to the representative of $B_j$, say $v$, and there is an edge from $v$ to a node $u$ in $B_j$'s parent bucket, and there is a path of length at most $2k-2$ between $u$ and $p_i$ in the parent bucket.
The following theorem summarizes our result.

\begin{theorem}
	\label{kDBST-thr}
	There exists a polynomial-time $(3k-2)$-approximation algorithm for the $k$-disjoint bottleneck spanning tree problem on points in a metric space. 
\end{theorem}

\paragraph{Remark.} The length $2k-2$ within each bucket of size $k$ is the best achievable. For example consider a tree rooted at a node $v$ with $k+1$ subtrees each is a path with  $k-1$ nodes. This tree has $k^2$ nodes in total which will be partitioned into $k$ buckets of size $k$. Since there are $k+1$ leaves at least two of them lie in the same bucket (by the pigeonhole principle), and thus the distance between them will be $2k-2$.

\section{The 2-GBST problem}

In this section we study the 2-GBST problem: Given a set $P$ of $n$ points in some metric space that is partitioned into clusters of size at most $2$, find a tree that contains exactly one point from each cluster and minimizes the largest edge length. First we prove the hardness of this problem and then present an approximation algorithm.

\begin{theorem}
	\label{2GBST-hardness}
	Unless P = NP, there is no polynomial-time algorithm that approximates the metric $2$-generalized bottleneck spanning tree problem by a factor better than 2.
\end{theorem}
\begin{proof}
	We use a reduction form the traditional 3-SAT problem: given  a boolean expression $E$ as the conjunction of clauses, each of which is the disjunction of three distinct literals (a variable or its negation), decide whether $E$ is satisfiable.\let\qed\relax\end{proof}
\vspace{-8pt}
	
	\begin{wrapfigure}{r}{0.3\textwidth}
		\begin{center}
			\vspace{-18pt}
			\includegraphics[width=.29\textwidth]{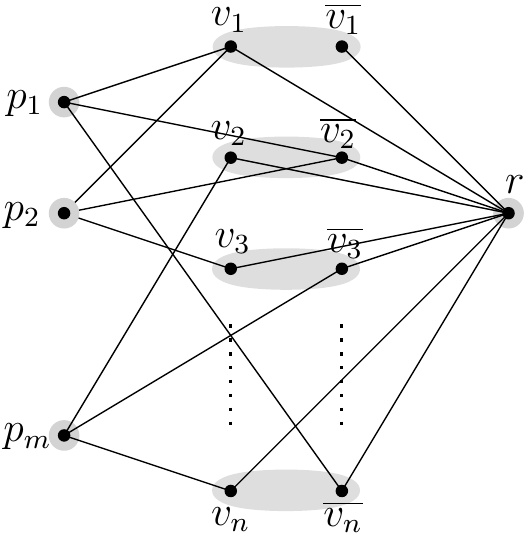}
			\label{notation-fig}
		\end{center}
		\vspace{-14pt}
	\end{wrapfigure}
	Given any instance of the 3-SAT problem consisting of an expression $E$ with $m$ clauses $C_1,\dots, C_m$ and $n$ variables $x_1,\dots,x_n$ we construct an instance of the 2-GBST problem consisting of a metric graph $G$ as follows (the vertices of $G$ represent points in a metric space). For each clause $C_j$ create a cluster with one vertex $p_j$. For each variable $x_i$ create a cluster with two {\em literal vertices} $v_i$ and $\overline{v_i}$ that  correspond to positive literal $x_i$ and negative literal $\overline{x_i}$, respectively. Create a cluster with one vertex $r$. To simplify our description we use vertices and their corresponding clauses or literals interchangeably. Connect each literal vertex, by edges of weight 1, to vertices $p_j$ of all clauses $C_j$ that they appear in. Connect $r$ to all literal vertices by edges of weight 1. All other edges of $G$ have weight 2. Notice that $G$ is a metric graph with $m+2n+1$ vertices. We show that $E$ is satisfiable if and only if $G$ has a generalized spanning tree with edges of weight 1. This would imply the statement of the theorem because (by contraposition) any approximation algorithm with factor less than 2 would give a tree with edges of weight 1, and thus could solve the 3-SAT problem.
	
	First suppose that $E$ is satisfiable, and consider a truth assignment of variables that satisfies $E$. We obtain a tree $T$ as follows. For the vertex set of $T$ we select $r$, all vertices $p_j$, and each $v_i$ (if $x_i$ is true) or $\overline{v_i}$ (if $x_i$ is false). For the edge set of $T$ we connect $r$ to every selected literal vertex, and we connect each $p_j$ to exactly one selected literal vertex that satisfies $C_j$. The tree $T$ is a feasible solution for the 2-GBST problem  on $G$ (as it contains exactly one vertex from each cluster) and all its edges have weight 1.
	
	For the other direction assume that $T$ is a generalized  spanning tree of $G$ with edges of weight 1. The tree $T$ should contain $r$ and all vertices $p_j$ because they are the only vertices in their clusters. For each $p_j$ only edges of $G$ that connect $p_j$ to literal vertices have weight 1. Thus each $p_j$ is connected to at least one literal vertex in $T$. Moreover $T$ contains exactly one vertex from each cluster $\{v_i,\overline{v_i}\}$ of literal vertices. Therefore, by setting $x_i$ as true (if $T$ contains $v_i$) or false (if $T$ contains $\overline{v_i}$) we obtain a satisfying assignment for $E$.
\qed
\vspace{8pt}

If in the proof of Theorem~\ref{2GBST-hardness} we replace all edge-weights of 2 with an arbitrary large constant, we obtain the following corollary.

\begin{corollary}
	\label{2GBST-hardness-cor}
	It is NP-hard to approximate the non-metric $2$-generalized bottleneck spanning tree problem within any constant factor.
\end{corollary}

If we were interested in generalized minimum spanning trees, then our reduction in the proof of Theorem~\ref{2GBST-hardness} would also give a short proof for the NP-hardness of the metric 2-GMST problem: It can be verified that $E$ is satisfiable if and only if $G$ has a generalized spanning tree of total weight $m+n$. We note the existence of (somewhat involved) proofs for the hardness of the Euclidean 2-GMST problem; see the thesis of Fraser \cite[page 140]{Fraser2012} (reduction from maximum 2-SAT), the paper of Ataei~\etal\cite{Ataei2018} (reduction from planar 3-SAT), and a recent result of Dey~\etal\cite{Dey2021} (reduction from maximum 2-SAT).

\subsection{A $3$-approximation for the $2$-GBST}

Here we present our 3-approximation algorithm for the 2-GBST problem on a set $P$ of $n$ points in a metric space that is partitioned into $m$ clusters $C_1,\dots,C_m$, each of size at most $2$. Notice that $n/2\leqslant m\leqslant n$. Let $\lambda^*$ be the bottleneck of a fixed optimal solution. 
In a nutshell, our algorithm works as follows. First we compute a tree $T_1$ that contains ``at least'' one point from each cluster and  its bottleneck is at most $\lambda^*$. Then we obtain a tree $T_2$ from $T_1$ that contains ``exactly'' one point from each cluster and its bottleneck is at most thrice $\lambda(T_1)$. Therefore 
$$\lambda(T_2)\leqslant 3\cdot \lambda(T_1)\leqslant 3\cdot \lambda^*,$$
which means that $T_2$ is a 3-approximate solution for the 2-GBST problem. In the rest of this section we show how to construct $T_1$ and $T_2$. Our algorithm for computing $T_2$ from $T_1$ is of independent interest.
The running time of our algorithm is dominated by the computation of a minimum spanning tree. 
The following theorem summarizes our result.

\begin{theorem}
	\label{2GBST-thr}
	There exists a polynomial-time 3-approximation algorithm for the $2$-generalized bottleneck spanning tree problem on points in a metric space. 
\end{theorem}

\subsubsection{Construction of $T_1$}

First we make an empty graph $G$ over the $n$ points of $P$. Then we add edges between the points of $G$ in a non-decreasing order of the distances, and stop as soon as $G$ has a connected component, say $C$, that contains at least one point from each cluster. All edges of $C$ are of length at most $\lambda^*$. Now we compute $T_1$ as an arbitrary spanning tree of $C$.

\paragraph{Remark.} When the running time is a concern, one can guess $\lambda^*$ in a binary search fashion to speed up the algorithm. Also, it is possible to compute $T_1$ as a subtree of the minimum spanning tree of $P$. In this case, the total running time is dominated by the computation of the minimum spanning tree; the details are removed as we are not concerned about the running time here.

\subsubsection{Construction of $T_2$}

In this section we prove the following theorem.  

\begin{theorem}
	\label{T1T2-thr}
	Given a tree $T_1$ and a partitioning of its nodes into clusters of size at most two, we can obtain a tree $T_2$ that contains exactly one node from each cluster and the length of its edges is at most $3$ in the metric of $T_1$. The upper bound $3$ is the best achievable.
\end{theorem}

First we show that the distance $3$ (in the metric of $T_1$) is the best achievable upper bound. Figure~\ref{lbound-fig} illustrates a tree $T_1$ as a path with eight nodes. The nodes of $T_1$ are partitioned into five clusters $\{a\}, \{b_1,b_2\},\allowbreak  \{c_1,c_2\},\allowbreak \{d_1,d_2\},\allowbreak \{e\}$. To obtain $T_2$ we have to choose points $a$ and $e$ because they are the only points in their clusters. Due to symmetry we may choose $b_1$ from cluster $\{b_1,b_2\}$. In this case if we do not choose $d_2$ then the distance of $e$ to its closest point in $T_2$ would be at least $3$, thus we may assume $d_2$ is chosen. In this setting, if we choose $c_1$ (as depicted in Figure~\ref{lbound-fig}) then the distance between $c_1$ and $d_2$ will be $3$, and if we choose $c_2$ then the distance between $b_1$ and $c_2$ will be $3$. Thus, in all cases we get an edge of length $3$ in $T_2$.

\begin{figure}[htb]
	\centering
	\includegraphics[width=.48\columnwidth]{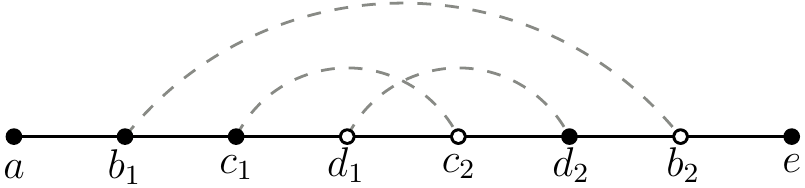}
	\caption{Illustration of the lower bound $3$. Dashed lines represent two nodes belonging to the same cluster. The black vertices are chosen for $T_2$.}
	\label{lbound-fig}
\end{figure}

Now we present an algorithm for obtaining $T_2$. Our algorithm consists of two phases: In the first phase we select the nodes of $T_2$ and in the second phase we define its edges. To select the nodes of $T_2$, we {\em visit} the nodes of $T_1$ in some order and {\em select} exactly one node from each cluster.  While visiting the nodes of $T_1$ we refer to an unvisited node by {\em open node}, to a visited node that is selected by {\em selected node}, and to a visited node that is not selected by {\em burned node}.

\vspace{10pt}				
\noindent{\bf Node selection.} See Figure~\ref{local-fig}(a) for an illustration of this phase. At the beginning all nodes of $T_1$ are open. First we visit and select all nodes of clusters of size one (which must be in $T_2$). Now we are going to select exactly one node from each cluster of size two. We root $T_1$ at an arbitrary node. Then we repeat the following process until all nodes of $T_1$ are visited. The process starts from an open node. At the beginning if the root is open then we start from the root, otherwise start from an arbitrary open node. In  Figure~\ref{local-fig}(a) the nodes are labeled by the order they have been visited; the nodes of clusters of size one (which are already visited) are labeled with 0s.

\begin{center}
	\fbox{%
		\parbox{.85\textwidth}{%
			\vspace{3pt}
			{\bf Process}: Let $a_1$ denote the starting open node (which belongs to a cluster of size two). Select $a_1$ and burn its twin say $a_2$. If the parent of $a_2$ is open then repeat the process starting from the parent. If the parent is not open (selected or burned) then check the children if $a_2$. If $a_2$ has some open child then repeat the process starting from an open child. If $a_2$ has no open child (or if $a_2$ does not have any child at all) then repeat the process starting from an arbitrary open node if such a node exists otherwise terminate the node selection phase.}
	}
\end{center}

\vspace{10pt}				
\noindent{\bf Defining edges.}  The node selection algorithm selects exactly one node from each cluster. At the end of the selection algorithm, every node is either selected or burned (there is no open node). 
We claim (proved below) that for each selected node $a$ at any level of $T_1$ (except for the root) there exists a selected node $b$ at a higher level such that the path between $a$ and $b$ in $T_1$ has at most three edges, i.e. the distance between $a$ and $b$ is at most 3 in the metric of $T_1$. For each selected node $a$, we add the edge $(a,b)$ to $T_2$. As each $a$ is connected to a node in a higher level, all nodes of $T_2$ are connected (via root) and hence it is a tree.

\begin{figure}[htb]
	\centering
	\setlength{\tabcolsep}{0in}
	$\begin{tabular}{cc}
		\multicolumn{1}{m{.68\columnwidth}}{\centering\includegraphics[width=.43\columnwidth]{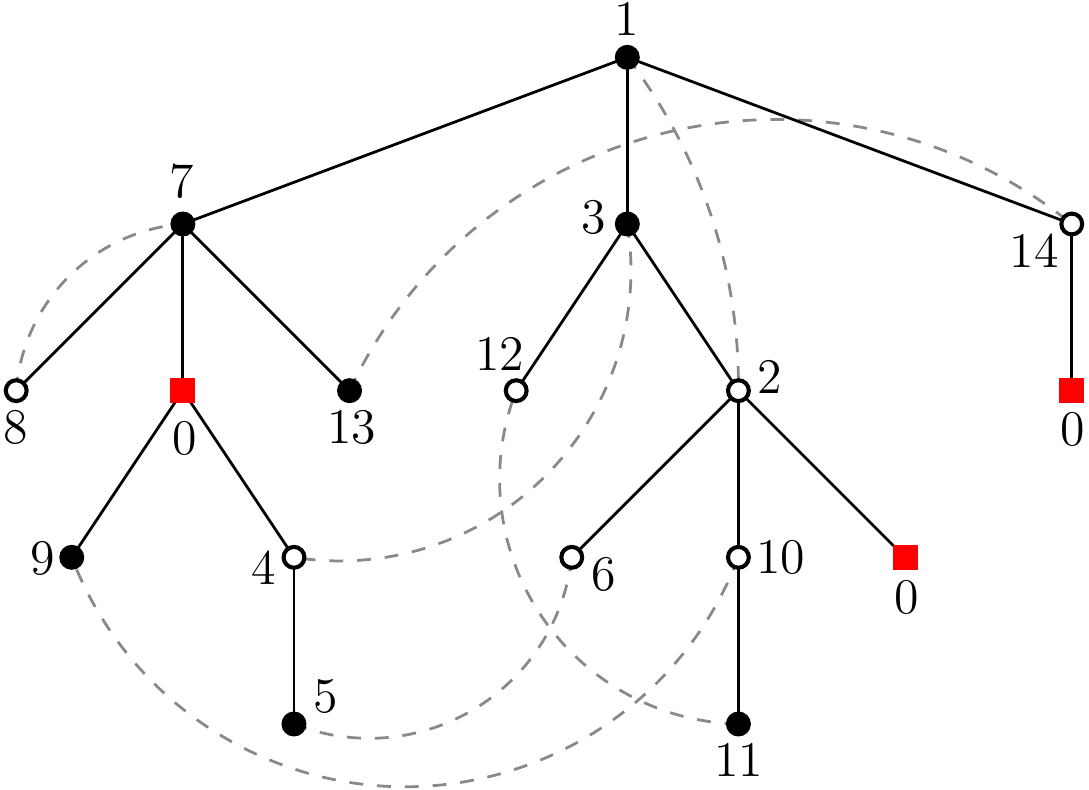}}
		&\multicolumn{1}{m{.32\columnwidth}}{\centering\includegraphics[width=.10\columnwidth]{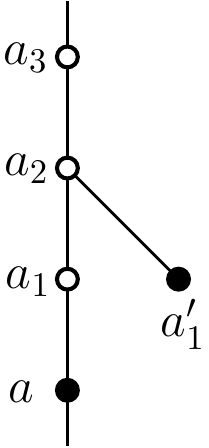}}
		\\
		(a)&(b)
	\end{tabular}$
	\caption{(a) Node selection (dashed lines represent two nodes in the same cluster): red squares (belong to clusters of size one) and black nodes (belong to clusters of size two) are selected whereas the white nodes (paired with black nodes) are burned. (b) Illustration for the edge length of $T_2$.}
	\label{local-fig}
\end{figure}

Now we verify the above claim. Let $a_1$ be the parent of $a$, as in Figure~\ref{local-fig}(b). If $a_1$ is selected then set $b=a_1$ and we are done. Assume that $a_1$ is burned. Let $a_2$ be the parent of $a_1$. If $a_2$ is selected then set $b=a_2$ and we are done. Assume that $a_2$ is also burned.
Notice that $a_2$ was burned before $a_1$ was, because otherwise the selection process would select $a_2$ right after burning $a_1$. Right after burning $a_2$ the process have checked the parent of $a_2$ which we denote by $a_3$. If $a_3$ was open then it would have been selected, and thus we set $b=a_3$ and we are done. If $a_3$ was burned then the process would have checked the children of $a_2$ and have selected a child $a'_1$ because $a_2$ had an open child which was $a_1$; this case is depicted in Figure~\ref{local-fig}(b). In this case we set $b=a'_1$ and we are done. The existence of $a_1$, $a_2$, and $a_3$ comes from the fact that the root of $T_1$ is a selected node.

\paragraph{Remark.} It might be tempting to use our $3$-approximation algorithm for the $2$-GBST problem to obtain a $3$-approximation for the $2$-DBST problem, say by coloring the selected nodes red and the burned nodes blue. This may not be an easy task because each time the process starts by selecting an {\em arbitrary} open node, these selected nodes could form a long path between burned nodes.

\section{The $k$-PBST problem}

Let $k\geqslant 2$ be an integer. In this section we study the $k$-PBST problem: Given $kn$ points in some metric space, find $k$ trees each containing exactly $n$ points and minimize the largest edge length (over all edges of the $k$ trees). 
First we prove the hardness of this problem. We assume that $n$ is at least 3, because if $n=2$ then the problem is equivalent to the bottleneck matching problem which can be solved in polynomial time.
Then we present an approximation algorithm for this problem. 


\begin{theorem}
	\label{2PBST-hardness}
	Unless P = NP, there is no polynomial-time algorithm that approximates the metric $k$-partition bottleneck spanning tree problem by a factor better than 2, for any $k\geqslant 2$.
\end{theorem}
\begin{proof}
	We use a reduction from the NP-hard problem of partitioning the vertex set of a graph $G=(V,E)$ into $k$ ($2\leqslant k\leqslant |V|/3$) equal-size subsets $V_1,\dots,V_k$ such that the induced subgraph by each $V_i$ is connected \cite{Dyer1985}. Let $G'$ be the complete edge-weighted graph obtained by adding edges to $G$ and then assigning weight 1 to every edge of $E$ and weight 2 to every other edge. Notice that $G'$ is a metric graph with $|V|$ vertices. It is easily seen that the partition problem on $G$ has a solution if and only if $G'$ contains $k$ equal-size spanning trees with edges of weight 1.
	The inapproximability claim follows because any approximation algorithm with factor less than 2 would give spanning trees with edges of weight 1, which would solve the partitioning problem on $G$.
\end{proof}

If in the proof of Theorem~\ref{2PBST-hardness} we replace all edge-weights of 2 with an arbitrary large constant, we obtain the following corollary.

\begin{corollary}
	\label{2PBST-hardness-cor}
	It is NP-hard to approximate the non-metric $k$-partition bottleneck spanning tree problem within any constant factor, for any $k\geqslant 2$.
\end{corollary}


\subsection{Approximating the $k$-PBST}

Now we present an $\alpha(k)$-approximation algorithm for the $k$-PBST problem, where $\alpha(k)=2$ for $k=2,3$ and $\alpha(k)=3$ for $k\geqslant 4$. In view of  Theorem~\ref{2PBST-hardness} the factor $2$ is the best achievable for $k=2,3$. 
Given $kn$ points in a metric space, we show how to construct $k$ trees each containing exactly $n$ points and their largest edge length is at most $\alpha(k)\cdot\lambda^*$, where $\lambda^*$ is the bottleneck of a fixed optimal solution. 

We start by computing a minimum spanning tree $T$ of all points. Let $e$ be a longest edge of $T$, that is $\lambda(T)=w(e)$. Let $T'$ and $T''$ be the two trees obtained by removing $e$ from $T$. If the number of nodes in $T'$ and in $T''$ are multiples of $n$, say $i\cdot n$ and $j\cdot n$ where $i+j=k$, then we recursively construct $i$ trees on the nodes of $T'$ and $j$ trees on the nodes of $T''$.

Assume that the number of nodes in $T'$ and $T''$ are not multiples of $k$. Then the optimal solution must have an edge between a node of $T'$ and a node of $T''$. The length of this edge is at least $w(e)$, and thus $\lambda^*\geqslant \lambda(T)$.
Then by Theorem~\ref{partitioning-thr} we obtain $k$ trees on the nodes of $T$ such that their edge lengths are at most $\alpha(k)\cdot\lambda(T)$. The following theorem summarizes our result in this section.

\begin{theorem}
	\label{PBST-thr}
	There exists a polynomial-time $\alpha$-approximation algorithm for the $k$-partition bottleneck spanning tree problem on points in a metric space where $\alpha=2$ for $k=2,3$ and $\alpha=3$ for $k\geqslant 4$. The approximation factor 2 for $k=2,3$ is the best achievable in polynomial time. 
\end{theorem}

\subsection{Balanced tree partitioning theorem}
In this section we prove the following theorem. 
We denote the number of nodes of a tree $T$ by $|T|$.

\begin{theorem}
	\label{partitioning-thr}
	Given a tree $T$ with $kn$ nodes we can obtain $k$ disjoint trees $T_1,\dots,T_k$ each containing exactly $n$ nodes of $T$ such that 	\begin{enumerate}
		\item If $k=2,3$ then the length of edges in each $T_i$ is at most $2$ in the metric of $T$. 
		\item If $k\geqslant 4$ then the length of edges in each $T_i$ is at most $3$ in the metric of $T$.
	\end{enumerate}
The upper bounds 2 and 3 for the edge lengths are the best achievable.
\end{theorem}

For the proof we first show that the upper bounds 2 and 3 are the best achievable. Then we present algorithms that achieve {\em desirable} trees $T_1,\dots,T_k$ with the claimed edge lengths. The lengths mentioned in our proof are in the metric of $T$.

\paragraph{Upper bounds.}
It is easily seen that the upper bound of 2 is the best achievable (for $k=2,3$) for example when $T$ is a star with 3 and 5 leaves, respectively.

\begin{wrapfigure}{r}{0.3\textwidth}
	\begin{center}
		\vspace{-23pt}
		\includegraphics[width=.28\textwidth]{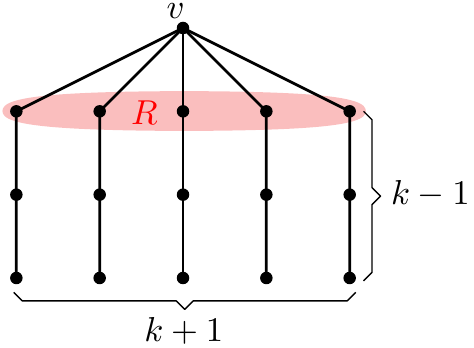}
		\label{notation-fig}
	\end{center}
	\vspace{-16pt}
\end{wrapfigure}
To verify that 3 is the best achievable upper bound (for $k\geqslant 4$) consider a tree $T$  rooted at a node $v$ with $k+1$ subtrees each is a path with $k-1$ nodes; see the figure to the right for $k=4$. The tree $T$ has $k^2$ nodes. Let $R$ be the set of $k+1$ nodes that are at distance 1 from $v$. Each node of $R$ {\em represents} a path connected to $v$. Now consider any set of $k$ disjoint trees $T_1,\dots,T_k$ each consisting of $k$ nodes of $T$. We show, by contradiction, that the length of an edge in some $T_i$ is at least 3. After a suitable relabeling assume that $v$ belongs to $T_1$. Then each tree $T_i$ with $i\in\{2,\dots,k\}$ should have nodes from at least two of the paths connected to $v$ because each path itself has $k-1$ nodes. In particular $T_i$ should contain the representatives of these paths because otherwise $T_i$ should have an edge of length at least $3$. Thus each $T_i$ contains at least two distinct nodes from $R$. This implies that $|R|\geqslant 2 (k-1)$. Combining this inequality with the fact that $|R|=k+1$, implies that $k\leqslant 3$ which is a contradiction.

\paragraph{Algorithm for $k\geqslant 4$.} Our algorithm for $k\geqslant 4$ uses the fact that the cube of $T$ is Hamiltonian. It is implied from a result of Karaganis \cite{Karaganis1968} and independently from a result of Lesniak \cite{Lesniak1973} that in polynomial time we can find a Hamiltonian path on nodes of $T$ with edges of length at most $3$. By cutting this path into $k$ equal-size pieces we obtain $k$ desired trees.

\paragraph{Remark.} One could simply obtain a 2-approximation if the {\em square}\footnote{The square of a graph $G$ has the same vertices as $G$, and has an edge between two distinct vertices if and only if there exists a path, with at most two edges, between them in $G$.} of $T$ has a Hamiltonian path. However, this property holds only for a very restricted class of trees called {\em horsetail} \cite{Radoszewski2011}.

\paragraph{Algorithm for $k=2$.}

\begin{figure}[htb]
	\centering
	\includegraphics[width=.8\columnwidth]{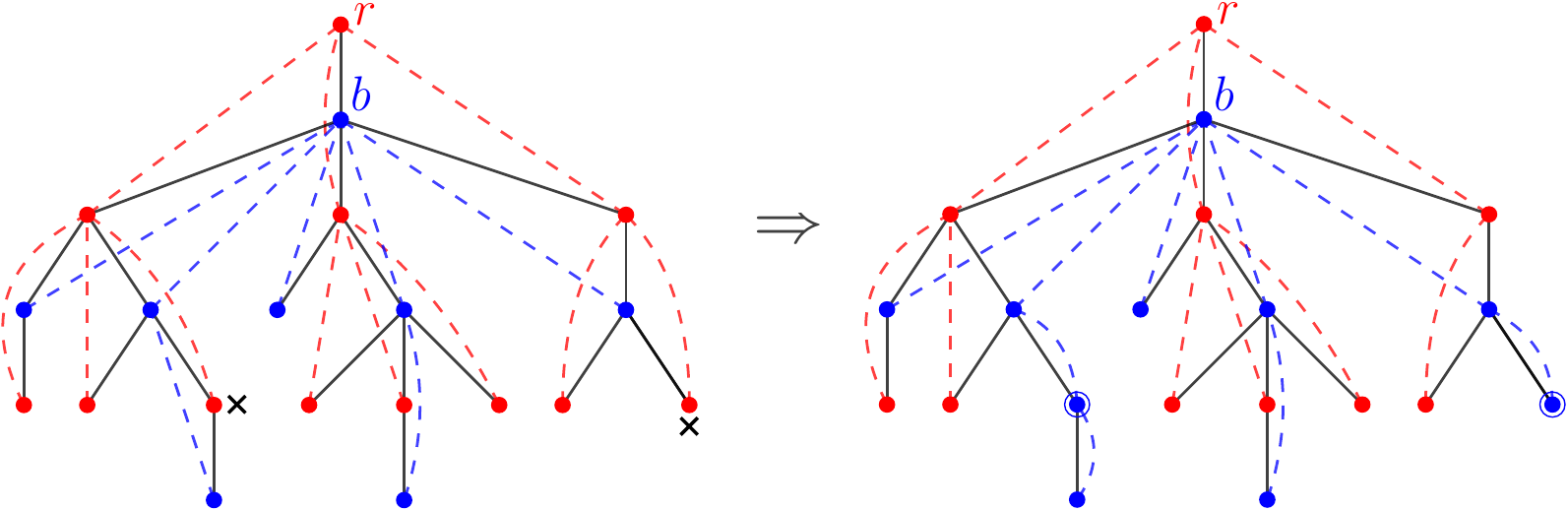}
	\caption{Obtaining trees $R$ (in red) and $B$ (in blue) from $T$ (in black).}
	\label{PBST-fig}
\end{figure}

We show how to find two disjoint trees $R$ and $B$ each containing exactly $n$ nodes of $T$ and the length of their edges is at most $2$.
To simplify our description we assume that the nodes of $R$ and $B$ are colored red and blue, respectively.

We root $T$ at a leaf $r$, as in Figure~\ref{PBST-fig}. Then $r$ has only one child which we denote by $b$. Assume that $r$ is at level 1, its child $b$ is at level 2, the children of $b$ are at level 3, and so on. Color all nodes at odd levels red and color all nodes at even levels blue. Compute a rooted tree $R$ on red points by connecting each red node to its grandparent, and compute a rooted tree $B$ on blue points by connecting each blue node to its grandparent, as in Figure~\ref{PBST-fig}. Notice that $R$ is rooted at $r$ and $B$ is rooted at $b$. Since each red node (resp. blue node) is connected to its grandparent, every edge of $R$ (resp. $B$) has length $2$.

If $R$ has $n$ nodes, so does $B$, and hence $\{R,B\}$ is a 2-approximate solution. If one tree, say $R$, has more than $n$ nodes, then we iteratively remove a leaf from $R$ until it is left with exactly $n$ nodes. We color the removed nodes of $R$ by blue; see Figure~\ref{PBST-fig}-right. Then we recompute the tree $B$ from the beginning by connecting each blue node to one of its parent and grandparent that is blue. Since no new edge is introduced in $R$, its edges still have length 2. Since each blue node is connected to its parent or grandparent (in $T$), the length of its edges is at most $2$. Therefore, the new trees $R$ and $B$ are desirable.

\paragraph{Remark.} It is easily seen that the above algorithm can be extended to obtain trees $R$ and $B$ of different sizes (as long as $|R|+|B|=|T|$) with the same upper bound of 2 on their edge lengths. 

\paragraph{Algorithm for $k=3$.} Notice that $T$ has $3n$ nodes.
We show how to find three disjoint trees $R$, $G$, and $B$ each containing exactly $n$ nodes of $T$ and the length of their edges is at most $2$.

We root $T$ at a leaf node. For each node $v$ in $T$, let $N(v)$ denote the number of nodes in the subtree rooted at $v$; the node $v$ is counted. Then we look at all nodes $v$ for which $N(v)$ is at least $n$. Among those, pick a node $v$ for which $N(v)$ is minimum. 
Then $N(v)$ is at least $n$ and each of its children has a subtree of size at most $n-1$. Observe that $v$ is not the root.

If $N(v)=n$ then we take the subtree rooted at $v$ as $R$, remove $R$ from $T$, and then obtain two trees $G$ and $B$ from the new tree $T$ (which now has $2n$ nodes) using our algorithm for $k=2$.

Assume that $N(v)>n$. Then $v$ has at least two children which we denote by $u_1, u_2,\dots, u_m$ where $m\geqslant 2$. Let $U_i$ denote the subtree rooted at $u_i$. Take the smallest index $j$ in $\{1,\dots,m\}$ for which $|U_1|+\dots +|U_j|\geqslant n$. Then $|U_1|+\dots+|U_{j-1}|<n$. Let $n'_1=n-(|U_1|+\dots+|U_{j-1}|)$. Let $U^v_j$ be the subtree consisting of $U_j$ and the node $v$ together with the edge connecting $v$ to $u_j$. We use our algorithm for $k=2$ to obtain from $U^v_j$ two trees $T'_j$ and $T''_j$ of sizes $n'_1$ and $|U^v_j|-n'_1=|U_j|+1-n'_1$, respectively, such that $T'_j$ is rooted at $u_j$, $T''_j$ is rooted at $v$, and their edge lengths are at most 2; see Figure~\ref{3PBST-fig}. Now we obtain $R$ by taking the trees  $U_1,\dots, U_{j-1}$, and $T'_j$ and interconnecting their roots to form one tree. Notice that $R$ has $n$ nodes and its edge lengths are at most 2. 
We remove the nodes of $R$ from $T$. We also remove all edges of $T$ that lie in $U_j$, and add the edges of $T''_j$ (which are of length at most 2) to $T$. Notice that $|T''_j|<n$ because it does not have $u_j$ (although it contains $v$). To obtain $G$ and $B$ we consider the following cases depending on the number $N(v)$ in the new tree $T$ which has $2n$ nodes:

\begin{figure}[htb]
	\centering
	\includegraphics[width=.85\columnwidth]{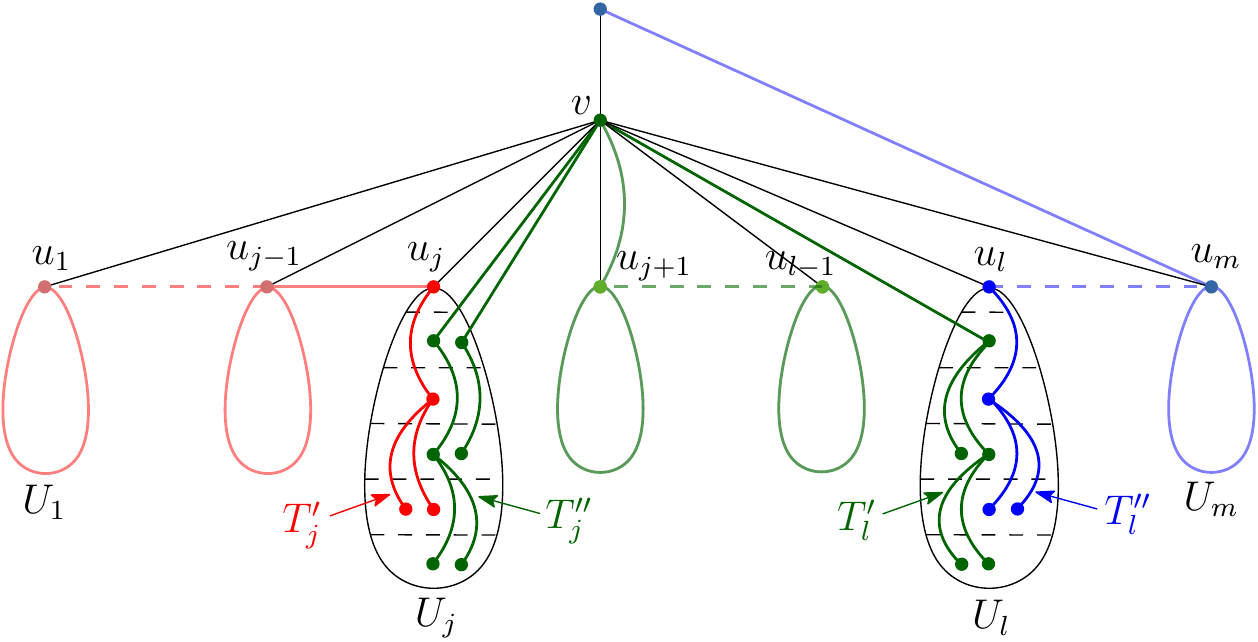}
	\caption{Obtaining trees $R$ (in red), $G$ (in green), and $B$ (in blue) from $T$ (in black). The trees $T'_j$, $T''_j$, $T'_l$, and $T''_l$ are shown with bold edges.}
	\label{3PBST-fig}
\end{figure}

\begin{itemize}
	\item $N(v)=n$. In this case we take the subtree rooted at $v$ as $G$, remove $G$ from $T$, and then take the resulting tree $T$ (which now has $n$ nodes) as $B$.

\item $N(v)<n$. We walk up the tree $T$ from $v$ and stop at the first node $w$ for which $N(w)\geqslant n$. We repeat the above process to obtain $G$ (which is now playing the role of $R$) but we denote the subtree of $w$ that contains $v$ by $U_1$. This ensures that the edges of $T''_j$ will appear in $G$ without getting longer. After obtaining $G$, the remaining part of $T$ will form the tree $B$.

\item $N(v)>n$. See Figure~\ref{3PBST-fig}. In this case we somehow repeat a procedure similar to what we did to obtain $R$. Let $l\in\{j+1,\dots,m\}$ be the smallest index for which $|T''_j|+|U_{j+1}|+\dots+|U_l|\geqslant n$. Notice that $U_{j+1}$ exists because $N(v)>n$. Then $|T''_j|+|U_{j+1}|+\dots+|U_{l-1}|< n$. Let $n'_2=n-(|T''_j|+|U_{j+1}|+\dots+|U_{l-1}|)+1$ (the addition of 1 will become clear shortly). Let $U^v_l$ be the subtree consisting of $U_l$ and the node $v$ together with the edge connecting $v$ to $u_l$ (notice that $v$ also belongs to $T''_j$). We use our algorithm for $k=2$ one more time to obtain from $U^v_l$ two trees $T'_l$ and $T''_l$ of sizes $n'_2$ and $|U^v_l|-n'_2=|U_l|+1-n'_2$, respectively, such that $T'_l$ is rooted at $v$, $T''_l$ is rooted at $u_l$ and their edge lengths are at most 2. Now we obtain $G$ by taking the trees $T''_j$, $U_{j+1},\dots, U_{l-1}$, and $T'_l$ and then interconnecting their roots to form one tree. The tree $G$ has $n$ nodes (without double counting $v$ which is in both $T''_j$ and $T'_l$) and its edge lengths are at most 2. 
We obtain the third tree, i.e. $B$, as follows. We remove the nodes of $G$ from $T$. By interconnecting the roots of $T''_l, U_{l+1},\dots, U_m$ together and then connecting $u_m$ to the parent of $v$ (which exists) we obtain the tree $B$.  

\paragraph{Remark.} To see why the above procedure cannot be extended to the case of $k=4$, assume that $N(v)>n$ after the removal of $G$ from $T$. As $v$ is already used for making $G$ we cannot reuse it to make another tree, and hence we will be forced to introduce longer edges.
\end{itemize}
\section{Conclusions}
A natural open problem is to improve the presented approximation ratios further. Most of our approximation ratios consider the largest edge length of the standard BST as the lower bound. A better lower bound for the largest edge length of an optimal solution (not the standard BST) would improve the approximation ratios.  It would be interesting to explore whether our algorithm for the $2$-GBST problem could be extended to an $O(k)$-approximation algorithm for the $k$-GBST problem. Also it would be interesting to verify whether the approximation ratio of $3$ for the $k$-PBST problem ($k\geqslant 4$) is tight, knowing that $2$ is a lower bound. 

\bibliographystyle{abbrv}
\bibliography{General-BST}
\end{document}